\documentclass[11pt]{article}

\usepackage[paper=letterpaper, margin=1in]{geometry}

\usepackage{ifthen}

\ifthenelse{\isundefined{\GenerateShortVersion}}
{
\def\LongVersion{}
\def\LongVersionEnd{}
\long\def\ShortVersion#1\ShortVersionEnd{}
}
{
\def\ShortVersion{}
\def\ShortVersionEnd{}
\long\def\LongVersion#1\LongVersionEnd{}
}

\newcommand{\Ignore}[1]{\ignorespaces}

\usepackage{amsmath, amssymb}

\usepackage{amsthm}

\usepackage{ifpdf}

\usepackage{authblk}

\ifpdf
\usepackage[pdftex]{graphicx}
\else
\usepackage[dvips]{graphicx}
\fi

\usepackage[%
  ocgcolorlinks,%
  linkcolor=blue,%
  filecolor=blue,%
  citecolor=blue,%
  urlcolor=blue]{hyperref}

\usepackage{enumitem}

\usepackage{xcolor}

\usepackage{mdframed}

\LongVersion 
\LongVersionEnd 

\ShortVersion 
\usepackage{times}
\ShortVersionEnd 

\ShortVersion 
\renewcommand{\paragraph}[1]{\par\noindent\textbf{#1}}
\ShortVersionEnd 

\newtheorem{theorem}{Theorem}[section]
\newtheorem{lemma}[theorem]{Lemma}
\newtheorem{observation}[theorem]{Observation}
\newtheorem{corollary}[theorem]{Corollary}
\newtheorem{proposition}[theorem]{Proposition}

\theoremstyle{definition}
\newtheorem{property}[theorem]{Property}
\theoremstyle{plain}

\Ignore{
\newtheorem{theorem}{Theorem}
\newtheorem{lemma}[theorem]{Lemma}
\newtheorem{observation}[theorem]{Observation}
\newtheorem{corollary}[theorem]{Corollary}
\newtheorem{proposition}[theorem]{Proposition}

\theoremstyle{definition}
\newtheorem{property}[theorem]{Property}
\theoremstyle{plain}
}

\Ignore{
\newtheorem{theorem}{Theorem}[section]
\newtheorem{lemma}{Lemma}[section]
\newtheorem{observation}{Observation}[section]
\newtheorem{corollary}{Corollary}[section]

\theoremstyle{definition}

\theoremstyle{plain}
}

\LongVersion 
\newenvironment{MathMaybe}[0]
{\begin{displaymath}\ignorespaces}
{\end{displaymath}}
\LongVersionEnd 
\ShortVersion 

\ShortVersionEnd 

\newcommand{\Integers}{\mathbb{Z}}
\newcommand{\Reals}{\mathbb{R}}
\newcommand{\Ex}{\mathbb{E}}
\renewcommand{\Pr}{\mathbb{P}}
\DeclareMathOperator{\Support}{supp}

\newcommand{\calH}{\mathcal{H}}
\newcommand{\bfp}{\mathbf{p}}
\newcommand{\bfq}{\mathbf{q}}

\title{Exploitation of Multiple Replenishing Resources with Uncertainty\footnote{%
This manuscript serves as a supplementary material for \cite{paper}.
The research conducted in this manuscript was partially supported by Lev-Zion Scholarship for Outstanding Ph.D.\ Students from Peripheral Areas, The Minerva Center for Movement Ecology, the Mexican Agency for International Development Cooperation (AMEXCID) and the European Research Council (ERC) under the European Union's Horizon 2020 research and innovation program (grant agreement No 648032).
 Y.E.\ was supported in part by an Israeli Science Foundation grant number 1016/17.}
}

\author[1]{Amos Korman}
\author[2]{Yuval Emek}
\author[1]{Simon Collet}
\author[3]{Aya Goldshtein}
\author[3]{Yossi Yovel}

\affil[1]{CNRS \& University of Paris, Paris, France.}
\affil[2]{Technion --- Israel Institute of Technology, Haifa, Israel.}
\affil[3]{Tel-Aviv University, Tel-Aviv, Israel.}

\begin{document}
\date{}
\maketitle
\begin{abstract}
We consider an optimization problem in which a (single) bat aims to exploit the nectar in a set of $n$ cacti with the objective of maximizing the expected total amount of nectar it drinks. Each cactus $i \in [n]$ is characterized by a parameter $r_{i} > 0$ that determines the rate in which nectar accumulates in $i$. In every round, the bat can visit one cactus and drink all the nectar accumulated there since its previous visit. Furthermore, competition with other bats, that may also visit some cacti and drink their nectar, is modeled by means of a stochastic process in which cactus $i$ is emptied in each round (independently) with probability $0 < s_i < 1$. Our attention is restricted to purely-stochastic strategies that are characterized by a probability vector
$(p_1, \ldots, p_n)$ determining the probability $p_i$ that the bat visits cactus $i$ in each round. We prove that for every $\epsilon > 0$, there exists a purely-stochastic strategy that approximates the optimal purely-stochastic strategy to within a multiplicative factor of $1 + \epsilon$, while exploiting only a small core of cacti.
Specifically, we show that it suffices to include at most
$\displaystyle \frac{2 (1 - \sigma)}{\epsilon \cdot \sigma}$ cacti in the core, where $\sigma = \min_{i \in [n]} s_{i}$. We also show that this upper bound on core size is asymptotically optimal as a core of a significantly smaller size cannot provide a $(1 + \epsilon)$-approximation of the optimal purely-stochastic strategy.
This means that when the competition is more intense (i.e., $\sigma$ is larger), a strategy based on exploiting smaller cores will be favorable.
\end{abstract}
\section{Model and Definitions}
\label{section:model}
Consider the \emph{hungry bat} problem defined over a single decision making
\emph{bat} and
$n \in \Integers_{> 0}$
\emph{cacti}.
At any given time, each cacti holds a finite non-negative amount of
\emph{nectar}.
The execution progresses in discrete \emph{rounds}
$t = 1, \dots, T$
so that the bat visits one cactus in each round.
If the bat visits cactus
$i \in [n]$
in round
$t \in [T]$,
then it collects the whole amount of nectar held in $i$ at that time and
the cactus is emptied.

Initially, the cacti hold no nectar and then, as the execution progresses,
they are filled with nectar according to the following process.
Each cactus
$i \in [n]$
is associated with real parameters
$r_{i} > 0$
and
$0 < s_{i} < 1$.
At the beginning of round
$t \in [T]$,
the amount of nectar in cactus $i$ is increased by $r_{i}$ units
and following that, a \emph{stealing event}, in which the cactus is emptied
from all the nectar held there, occurs w.p.\
$s_{i}$.

Formally, for cactus
$i \in [n]$
and round
$t \in [T]$,
let
$C_{i}(t) \in \Reals_{\geq 0}$
be the random variable that captures the amount of nectar held in cactus $i$
at the beginning of round $t$ and let
$B_{i}(t) \in \Reals_{\geq 0}$
be the random variable that captures the amount of nectar collected by the bat
from cactus $i$ in round $t$.
Conditioning on
$C_{i}(t) = z$,
we define the random variable
\[
C'_{i}(t)
\, = \,
\begin{cases}
z + r_{i}, & \text{w.p.\ } 1 - s_{i} \\
0, & \text{w.p.\ } s_{i}
\end{cases} \, .
\]
and based on that, the random variables $B_{i}(t)$ and
$C_{i}(t + 1)$
are set as
\[
B_{i}(t)
\, = \,
\begin{cases}
C'_{i}(t), &
\text{if the bat visits cactus $i$ in round $t$} \\
0, & \text{otherwise}
\end{cases}
\]
and
\[
C_{i}(t + 1)
\, = \,
C'_{i}(t) - B_{i}(t) \, .
\]

The bat wishes to maximize the total nectar amount
$\sum_{i \in [n]} \sum_{t \in [T]} B_{i}(t)$
it collects from all cacti throughout the execution.
This expression is a random variable, subject to the probabilistic nectar
stealing events and the coin tosses of the bat's strategy (if any), and the
goal is to maximize its expected value.

The bat's strategy is called \emph{purely-stochastic} if there exists some
probability vector
$\bfp = (p_{1}, \dots, p_{n})$
such that the bat visits cactus
$i \in [n]$
in round $t$ with probability $p_{i}$ for every
$t \in [T]$,
independently of the cacti visited in other rounds.
To emphasize that cactus
$i \in [n]$
is visited with probability $p_{i}$ under the purely-stochastic strategy
$\bfp$, we often write $B_{i}^{p_{i}}(\cdot)$ instead of $B_{i}(\cdot)$.
Unless stated otherwise, all strategies considered hereafter are
purely-stochastic.

Given a hungry bat instance
$\calH = \langle n, \{ r_{i} \}_{i \in [n]}, \{ s_{i} \}_{i \in [n]} \rangle$,
our aim is to design a (purely-stochastic) strategy
$\bfp = (p_{1}, \dots, p_{n})$
that maximizes the expected per-round amount of nectar collected by a bat that
follows $\bfp$ over $T$ rounds as
$T \rightarrow \infty$,
denoted by
\[
\mathcal{B}_{\calH}(\bfp)
\, = \,
\lim_{T \rightarrow \infty}
\Ex \left(
\frac{1}{T} \sum_{i \in [n]} \sum_{t = 1}^{T} B_{i}^{p_{i}}(t)
\right) \, .
\]
We refer to a strategy $\bfp$ that maximizes $\mathcal{B}_{\calH}(\bfp)$ as an
\emph{optimal} strategy for $\calH$.

Our contribution regarding the hungry bat problem is two-fold.
First, we develop some non-trivial insights on the structure of the optimal
strategies.

\begin{theorem} \label{theorem:optimal-strategy}
Every $n$-cacti hungry bat instance $\calH$ admits a unique optimal strategy
$\bfp^{*} = (p^{*}_{1}, \dots, p^{*}_{n})$
that can be constructed by a simple computationally efficient method.
\end{theorem}

Following that, we prove that although the optimal strategy $\bfp^{*}$ may
have a large support
$\Support(\bfp^{*}) = \{ i \in [n] \mid p^{*}_{i} > 0 \}$,
the bat can in fact collect almost all the nectar amount guaranteed by
$\bfp^{*}$ while visiting only a small subset $S$ of the cacti, referred to as
a \emph{core}.

\begin{theorem} \label{theorem:small-core}
Consider a hungry bat instance
$\calH = \langle n, \{ r_{i} \}_{i \in [n]}, \{ s_{i} \}_{i \in [n]} \rangle$
with optimal strategy $\bfp^{*}$
and let
$\sigma = \min_{i \in [n]} s_{i}$.
For every
$0 < \varepsilon < 1$,
there exists a core
$S = S(\varepsilon) \subseteq [n]$
of size
\[
|S|
\, \leq \,
\frac{2 (1 - \sigma)}{\varepsilon \cdot \sigma}
\]
and a strategy $\bfp$ for $\calH$ whose support is
$\Support(\bfp) = S$
such that
$\mathcal{B}_{\calH}(\bfp)
\geq
(1 - \varepsilon) \cdot \mathcal{B}_{\calH}(\bfp^{*})$.
\end{theorem}

A parameter that turns out to play a major role in constructing good
strategies for the hungry bat instance $\calH$ is
\[
\chi_{i}
\, = \,
\frac{r_{i} \cdot (1 - s_{i})}{s_{i}} \, ,
\]
$i \in [n]$.
Indeed, the support $\Support(\bfp^{*})$ of the optimal strategy $\bfp^{*}$
promised in Theorem~\ref{theorem:optimal-strategy} and the small core $S$
promised in Theorem~\ref{theorem:small-core} consist of cacti
$i \in [n]$
that admit a larger $\chi_{i}$ than any cacti not included in
$\Support(\bfp^{*})$ and $S$, respectively.
As we show later, the parameter $\chi_{i}$ captures the
amount of nectar held in cactus $i$ after the bat has not visited it for $t$
rounds as
$t \rightarrow \infty$.

Theorem~\ref{theorem:optimal-strategy} is established in
Section~\ref{section:solving-optimization-problem} and
Theorem~\ref{theorem:small-core} is established in
Section~\ref{section:small-core}, where we also prove that the promised bound
on the size of the core is asymptotically tight.
We start by analyzing the expected amount of nectar collected from a single
cactus in Section~\ref{section:single-cactus}.

\section{Single Cactus}
\label{section:single-cactus}
Throughout this section, we fix some cactus $i$ with parameters
$r = r_{i}$
and
$s = s_{i}$
and a (purely-stochastic) strategy
$\bfp = (p_{1}, \dots, p_{n})$
with probability
$p = p_{i}$
to visit cactus $i$.
Our goal is to analyze the expected per-round amount of nectar collected by
the bat from cactus $i$, i.e.,
\begin{equation} \label{equation:single-cactus}
\lim_{T \rightarrow \infty}
\Ex \left( \frac{1}{T} \sum_{t = 1}^{T} B_{i}^{p}(t) \right) \, .
\end{equation}
The following observation allows us to decouple between the nectar stealing
events and the visits of the bat.

\begin{observation} \label{observation:gap-determines-amount}
Suppose that the bat visits cactus $i$ in round $t_{0}$.
If the next time the bat visits cactus $i$ is in round
$t_{0} + x$,
then the expected amount of nectar it collects during this visit is
\[
\frac{1 - s}{s} (1 - (1 - s)^{x}) \cdot r \, ,
\]
where the expectation is over the nectar stealing events.
\end{observation}
\begin{proof}
For
$x \geq 1$,
let $\zeta_{x}$ be the random variable that captures the amount of nectar held
in cactus $i$ at the end of round
$t_{0} + x$
assuming that the bat does not visit $i$ in any round
$t_{0} < t \leq t_{0} + x$.
By definition,
$\Ex(\zeta_{1}) = (1 - s) \cdot r$
and
$\Ex(\zeta_{x + 1} \mid \zeta_{x}) = (1 - s) (\zeta_{x} + r)$,
hence
$\Ex(\zeta_{x})$
obeys the recursion
\[
\Ex(\zeta_{x})
\, = \,
\Ex \left( \Ex \left( \zeta_{x} \mid \zeta_{x - 1} \right) \right)
\, = \,
\begin{cases}
(1 - s) \cdot r, & \text{if } x = 1 \\
(1 - s) \left( \Ex(\zeta_{x - 1}) + r \right), & \text{if } x > 1
\end{cases} \, .
\]
It follows by induction on $x$ that
$\Ex(\zeta_{x}) = \sum_{j = 1}^{x} (1 - s)^{j} \cdot r$
which establishes the assertion as
$\sum_{j = 1}^{x} (1 - s)^{j} = (1 - s) \frac{1 - (1 - s)^{x}}{s}$.
\end{proof}

Notice that
\[
\lim_{x \rightarrow \infty} \frac{1 - s}{s} (1 - (1 - s)^{x}) \cdot r
\, = \,
\frac{r \cdot (1 - s)}{s}
\, = \,
\chi_{i} \, .
\]
Observation~\ref{observation:gap-determines-amount} is therefore consistent
with the informal explanation provided in Section~\ref{section:model} for the
parameter $\chi_{i}$.

To facilitate the analysis, we view the execution presented in
Section~\ref{section:model} as a stochastic process, denoted
hereafter by $\pi$, and introduce three alternative stochastic processes,
showing that they are equivalent to $\pi$ in terms of the expression in
(\ref{equation:single-cactus}).
First, stochastic process $\pi_{1}$ is defined similarly to $\pi$ except that
the amounts of nectar collected by the bat from cactus $i$ during its visits
are determined \emph{deterministically} by the number of rounds that have
passed since the previous visit, thus neutralizing the probabilistic effect of
the nectar stealing events.
Formally, given that the bat visits cactus $i$ in rounds $t_{0}$ and
$t_{1} > t_{0}$,
with no visits in between these two rounds, we set the (deterministic) amount
of nectar collected by the bat from cactus $i$ in round $t_{1}$ to be
\begin{equation} \label{equation:single-round-single-cactus}
B_{i}^{p}(t_{1}) = \frac{1 - s}{s} (1 - (1 - s)^{t_{1} - t_{0}}) \cdot r \, .
\end{equation}
Since the nectar stealing events are independent from the bat's visits (under
a purely-stochastic strategy),
Observation~\ref{observation:gap-determines-amount} ensures that in total,
over the whole time window
$t = 1, \dots, T$,
the expected amount of nectar collected by the bat from cactus $i$ under
$\pi_{1}$ is equal to the expected amount of nectar collected by the bat from
cactus $i$ under $\pi$.

To avoid dealing with the corner cases that occur at the beginning and at the
end of the time window
$t = 1, \dots, T$,
we introduce stochastic process $\pi_{2}$ defined to be an extension of
$\pi_{1}$ so that it lasts from
$t = -\infty$
to
$t = \infty$.
This is formally defined over a (two-sided) infinite sequence
$\{ V_{t} \}_{t = -\infty}^{\infty}$
of i.i.d.\ Bernoulli random variables with success probability $p$ so that
$V_{t}$,
$t \in \Integers$,
is an indicator for the event that the bat visits cactus $i$ in round $t$.
In accordance with (\ref{equation:single-round-single-cactus}), the nectar
amount $B_{i}^{p}(t_{1})$ collected by the bat from cactus $i$ in round
$t_{1}$ is now defined to be $0$ if
$V_{t_{1}} = 0$
(i.e., the cactus is not visited in round $t_{1}$);
and
$\frac{1 - s}{s} (1 - (1 - s)^{t_{1} - t_{0}}) \cdot r$
if
$V_{t_{1}} = V_{t_{0}} = 1$
and
$V_{t} = 0$
for every
$t_{0} < t < t_{1}$.

We emphasize that although $\pi_{2}$ is assumed to last ad infinitum, we are
still interested in the expected nectar amount collected during the (finite)
time window
$t = 1, \dots, T$
as defined in (\ref{equation:single-cactus}).
During this time window, the nectar amounts collected under $\pi_{2}$ may
differ from those of $\pi_{1}$ only in the first and last visits of the bat to
cactus $i$.
Since the amount of nectar collected in a single visit (under $\pi_{1}$ and
$\pi_{2}$) is up-bounded by
$\frac{r \cdot (1 - s)}{s}$
and since we consider the time window
$t = 1, \dots, T$
as
$T \rightarrow \infty$,
it follows that $\pi_{2}$ is equivalent to $\pi_{1}$ in terms of the limit in
(\ref{equation:single-cactus}).

Having established that, we construct stochastic process $\pi_{3}$ from
$\pi_{2}$ by equally dividing (virtually, only for the sake of the analysis)
the nectar amount $B_{i}^{p}(\cdot)$, collected by the bat from cactus $i$
when $i$ is visited, over the rounds that have passed since the bat's previous
visit to $i$.
Formally, for
$t \in \Integers$,
let $Y^{<}(t)$ be the random variable that takes on the largest
$y < t$
such that
$V_{y} = 1$,
let $Y^{\geq}(t)$ be the random variable that takes on the smallest
$y \geq t$
such that
$V_{y} = 1$,
and let
$X(t) = Y^{\geq}(t) - Y^{<}(t)$
be the random variable that counts the number of rounds between the previous
visit of cactus $i$ and the next one.
Based on that, the nectar amount collected by the bat from cactus $i$ in round
$t$ under $\pi_{3}$ is defined to be
\[
\widehat{B}_{i}^{p}(t)
\, = \,
\frac{(1 - s) (1 - (1 - s)^{X(t)}) \cdot r}{s \cdot X(t)} \, .
\]

Notice that under $\pi_{3}$, the bat collects a positive amount
$\widehat{B}_{i}(t)$
of nectar from cactus $i$ in every round
$t \in \Integers$
even if $i$ is not visited in round $t$.
Moreover, the amount of nectar $\widehat{B}_{i}(t)$ collected by the bat in a
given round $t$ is now determined in hindsight (upon the bat's next visit to
cactus $i$), rather than at real time.
This does not impose any obstacle as $\pi_{3}$ is employed only for the sake
of the analysis.
We are now ready to state the following two observations.

\begin{observation} \label{observation:virtual-equivalent-to-original}
Stochastic process $\pi_{3}$ is equivalent to the original stochastic process
$\pi$ in the sense that
\[
\lim_{T \rightarrow \infty}
\Ex \left( \frac{1}{T} \sum_{t = 1}^{T} \widehat{B}_{i}^{p}(t) \right)
\, = \,
\lim_{T \rightarrow \infty}
\Ex \left( \frac{1}{T} \sum_{t = 1}^{T} B_{i}^{p}(t) \right) \, .
\]
\end{observation}

\begin{observation} \label{observation:all-rounds-are-the-same}
Under stochastic process $\pi_{3}$, the random variables
$\widehat{B}_{i}^{p}(t)$,
$t \in \Integers$,
are identically distributed.
In particular,
\[
\lim_{T \rightarrow \infty}
\Ex \left( \frac{1}{T} \sum_{t = 1}^{T} \widehat{B}_{i}^{p}(t) \right)
\, = \,
\Ex \left( \widehat{B}_{i}^{p}(t) \right)
\]
for any round
$t \in \Integers$.
\end{observation}

Based on Observations \ref{observation:virtual-equivalent-to-original} and
\ref{observation:all-rounds-are-the-same}, our goal in the remainder of this
section is to analyze
$\Ex(\widehat{B}_{i}^{p}(t))$
for an arbitrary round
$t \in \Integers$.
The following lemma plays a key role in this task.

\begin{lemma} \label{lemma:length-geometric-phase}
The random variable $X(t)$,
$t \in \Integers$,
satisfies
\[
\Pr \left( X(t) = x \right)
\, = \,
x p \cdot (1 - p)^{x - 1} \cdot p \, .
\]
\end{lemma}
\begin{proof}
Fix some
$1 \leq z \leq x$.
The event
$Y^{<}(t) = t - z$
occurs if and only if
$V_{t - z} = 1$
and
$V_{t - z + 1} = V_{t - z + 2} = \cdots = V_{t - 1} = 0$,
hence
\[
\Pr \left( Y^{<}(t) = t - z \right)
\, = \,
p \cdot (1 - p)^{z - 1} \, .
\]
Conditioning on
$Y^{<}(t) = t - z$,
the event
$X(t) = x$
occurs if and only if
$V_{t - z + x} = 1$
and
$V_{t} = V_{t + 1} = \cdots = V_{t - z + x - 1} = 0$,
hence
\[
\Pr \left( X(t) = x \mid Y^{<}(t) = t - z \right)
\, = \,
(1 - p)^{x - z} \cdot p \, .
\]
Put together, we conclude that
\begin{align*}
\Pr \left( X(t) = x \right)
\, = \, &
\sum_{z = 1}^{x}
\Pr \left( Y^{<}(t) = t - z \right) \cdot
\Pr \left( X(t) = x \mid Y^{<}(t) = t - z \right) \\
= \, &
\sum_{z = 1}^{x}
p \cdot (1 - p)^{z - 1} \cdot (1 - p)^{x - z} \cdot p \\
= \, &
x p \cdot (1 - p)^{x - 1} \cdot p \, ,
\end{align*}
thus establishing the assertion.
\end{proof}

We are now ready to analyze the expected nectar amount
$\Ex(\widehat{B}_{i}(t))$
collected by the bat from cactus $i$ in an arbitrary round
$t \in \Integers$.
To this end, we develop
\begin{align*}
\Ex \left( \widehat{B}_{i}(t) \right)
\, = \, &
\sum_{x = 1}^{\infty}
\Ex \left( \widehat{B}_{i}(t) \mid X(t) = x \right)
\cdot
\Pr \left( X(t) = x \right) \\
= \, &
\sum_{x = 1}^{\infty}
\frac{(1 - s) \left( 1 - (1 - s)^{x} \right) \cdot r}{s \cdot x}
\cdot
x p \cdot (1 - p)^{x - 1} \cdot p \\
= \, &
\frac{(1 - s) \cdot r \cdot p^{2}}{s}
\cdot
\sum_{x = 1}^{\infty} \left( 1 - (1 - s)^{x} \right) (1 - p)^{x - 1} \\
= \, &
\frac{(1 - s) \cdot r \cdot p^{2}}{s}
\cdot
\left(
\sum_{x = 1}^{\infty} (1 - p)^{x - 1}
\, - \,
(1 - s) \cdot \sum_{x = 1}^{\infty} \left[ (1 - s) (1 - p) \right]^{x - 1}
\right) \\
= \, &
\frac{(1 - s) \cdot r \cdot p^{2}}{s}
\cdot
\left(
\frac{1}{p}
\, - \,
\frac{1 - s}{1 - (1 - s) (1 - p)}
\right) \\
= \, &
\frac{(1 - s) \cdot r \cdot p^{2}}{s}
\cdot
\left(
\frac{1}{p}
\, - \,
\frac{1}{\frac{1}{1 - s} - (1 - p)}
\right) \, .
\end{align*}
Using the expansions
$\frac{1}{1 - z} = \sum_{j = 0}^{\infty} z^{j}$
and
$\frac{z}{1 - z} = \sum_{j = 1}^{\infty} z^{j}$
that hold for any
$0 < z < 1$,
we conclude that
\begin{align*}
\Ex \left( \widehat{B}_{i}(t) \right)
\, = \, &
\frac{(1 - s) \cdot r \cdot p^{2}}{s}
\cdot
\left(
\frac{1}{p}
\, - \,
\frac{1}{\sum_{j = 0}^{\infty} s^{j} - (1 - p)}
\right) \\
= \, &
\frac{(1 - s) \cdot r \cdot p}{s}
\cdot
\left(
1
\, - \,
\frac{p}{p + \sum_{j = 1}^{\infty} s^{j}}
\right) \\
= \, &
\frac{(1 - s) \cdot r \cdot p}{s}
\cdot
\frac{\sum_{j = 1}^{\infty} s^{j}}{p + \sum_{j = 1}^{\infty} s^{j}} \\
= \, &
\frac{r \cdot p}{p + \sum_{j = 1}^{\infty} s^{j}}
\cdot
(1 - s)
\cdot
\sum_{j = 0}^{\infty} s^{j} \\
= \, &
\frac{r \cdot p}{p + \frac{s}{1 - s}}
\cdot
\left(
\sum_{j = 0}^{\infty} s^{j}
-
\sum_{j = 1}^{\infty} s^{j}
\right) \\
= \, &
\frac{r \cdot p}{p + \frac{s}{1 - s}} \, .
\end{align*}

We conclude with the following lemma.

\begin{lemma} \label{lemma:b-function}
For
$i \in [n]$,
the expected per-round amount of nectar collected by the bat from cactus $i$
under a strategy
$\bfp = (p_{1}, \dots, p_{n})$
with probability
$p = p_{i}$
to visit cactus $i$ satisfies
\[
\lim_{T \rightarrow \infty}
\Ex \left( \frac{1}{T} \sum_{t = 1}^{T} B_{i}^{p}(t) \right)
\, = \,
b_{i}(p) \, ,
\]
where
$b_{i} : [0, 1] \rightarrow \Reals_{> 0}$
is the function defined by setting
\[
b_{i}(p)
\, = \,
r_{i} \cdot \frac{p}{p + \frac{s_{i}}{1 - s_{i}}} \, .
\]
\end{lemma}

Notice that the first and second derivatives
\begin{equation} \label{equation:b-function-derivatives}
\frac{d}{d p} b_{i}(p)
\, = \,
r_{i} \cdot
\frac{\frac{s_{i}}{1 - s_{i}}}{\left( p + \frac{s_{i}}{1 - s_{i}} \right)^{2}}
\qquad \text{and} \qquad
\frac{d^{2}}{d p^{2}} b_{i}(p)
\, = \,
-2 r_{i} \cdot
\frac{\frac{s_{i}}{1 - s_{i}}}{\left( p + \frac{s_{i}}{1 - s_{i}} \right)^{3}}
\end{equation}
of $b_{i}(p)$ are strictly positive and strictly negative, respectively, in
the interval
$[0, 1]$.
Therefore, the function $b_{i}(p)$ is concave and its first derivative
$\frac{d}{d p} b_{i}(p)$
is monotonically (strictly) decreasing in
$[0, 1]$.

\section{Solving the Optimization Problem}
\label{section:solving-optimization-problem}
Consider an instance
$\calH = \langle n, \{ r_{i} \}_{i \in [n]}, \{ s_{i} \}_{i \in [n]} \rangle$
of the hungry bat problem.
Recall that the function
$\mathcal{B}_{\calH} : \Delta(n) \rightarrow \Reals_{> 0}$
maps every (purely-stochastic) strategy
$\bfp = (p_{1}, \dots, p_{n})$
in the $n$-dimensional probability simplex $\Delta(n)$ to
\[
\mathcal{B}_{\calH}(\bfp)
\, = \,
\lim_{T \rightarrow \infty}
\Ex \left(
\frac{1}{T} \sum_{i \in [n]} \sum_{t = 1}^{T} B_{i}^{p_{i}}(t)
\right)
\, = \,
\sum_{i \in [n]} b_{i}(p_{i}) \, ,
\]
where the second transition is due to Lemma~\ref{lemma:b-function}.
This function is continuous since the functions
$b_{i}(p_{i})$,
$i \in [n]$,
are continuous, therefore $\calH$ admits an optimal strategy $\bfp^{*}$ by the
compactness of the probability simplex.
In this section, we show that $\bfp^{*}$ is unique and develop a simple and
computationally efficient algorithm for constructing it, thus establishing
Theorem~\ref{theorem:optimal-strategy}.
We start by stating a necessary and sufficient condition for the
optimality of a strategy.

\begin{lemma} \label{lemma:optimality-condition}
A strategy
$\bfp = (p_{1}, \dots, p_{n})$
is optimal for $\calH$ if and only if
\[
\frac{d}{d p} b_{i}(p_{i})
\, = \,
\max_{j \in [n]} \frac{d}{d p} b_{j}(p_{j})
\]
for every
$i \in \Support(\bfp)$.
\end{lemma}
\begin{proof}
Suppose that
$\bfp = (p_{1}, \dots, p_{n})$
is a maximum point of the function
$\mathcal{B}_{\calH}(\cdot)$
and assume towards contradiction that there exists some
$i \in \Support(\bfp)$
and
$j \in [n]$
such that
$\frac{d}{d p} b_{j}(p_{j})
>
\frac{d}{d p} b_{i}(p_{i})$.
Employing the continuity of $b_{i}(p)$ and $b_{j}(p)$ in
$[0, 1]$
and the fact that $p_{i}$ is strictly positive, we conclude that one can
increase $\mathcal{B}_{\calH}(\bfp)$ by shifting mass from $p_{i}$
to $p_{j}$, thus deriving a contradiction to the maximality of $\bfp$.

In the converse direction, let
$\mu = \max_{i \in [n]} \frac{d}{d p} b_{i}(p_{i})$
and suppose that
$i \in \Support(\bfp)$
implies that
$\frac{d}{d p} b_{i}(p_{i}) = \mu$.
Recalling that the functions $b_{i}(p_{i})$,
$i \in [n]$,
are concave, we conclude that the function $\mathcal{B}_{\calH}(\bfp)$ is
concave too.
We prove that $\bfp$ is a local maximum of $\mathcal{B}_{\calH}(\cdot)$, hence
it is also a global maximum due to the function's concavity.
We do so by arguing that
\[
(\bfq - \bfp) \, \nabla \mathcal{B}_{\calH}(\bfp)
\, \leq \,
0
\]
for every probability vector
$\bfq = (q_{1}, \dots, q_{n})$,
where
$\nabla \mathcal{B}_{\calH}(\bfp)$
is the (column) gradient vector of $\mathcal{B}_{\calH}(\cdot)$ at
$\bfp$.
To this end, we develop
\begin{align*}
(\bfq - \bfp) \, \nabla \mathcal{B}_{\calH}(\bfp)
\, = \, &
\sum_{i \in [n]}
(q_{i} - p_{i}) \cdot \frac{d}{d p} b_{i}(p_{i}) \\
= \, &
\sum_{i \in \Support(\bfp)}
(q_{i} - p_{i}) \cdot \mu
+
\sum_{i \in [n] - \Support(\bfp)}
q_{i} \cdot \frac{d}{d p} b_{i}(p_{i}) \\
\leq \, &
\sum_{i \in \Support(\bfp)}
(q_{i} - p_{i}) \cdot \mu
+
\sum_{i \in [n] - \Support(\bfp)}
q_{i} \cdot \mu \\
= \, &
\mu - \mu
\, = \,
0 \, ,
\end{align*}
thus establishing the assertion.
\end{proof}

To avoid cumbersome notation, we subsequently reorder the cacti according to
their derivatives at
$p = 0$
so that
\begin{equation} \label{equation:rename-cacti}
\frac{d}{d p} b_{i}(0)
\, \geq \,
\frac{d}{d p} b_{i + 1}(0)
\end{equation}
for every
$1 \leq i \leq n - 1$.
Using (\ref{equation:b-function-derivatives}), it is interesting to point out
that the derivative at
$p = 0$
is exactly
\[
\frac{d}{d p} b_{i}(0)
\, = \,
\frac{r_{i} \cdot (1 - s_{i})}{s_{i}}
\, = \,
\chi_{i} \, ,
\]
i.e., the expected nectar amount held in cactus $i$ after it has not been
visited by the bat for $t$ rounds as
$t \rightarrow \infty$.
The following observation plays a key role in the design of our algorithm.

\begin{observation} \label{observation:prefix}
If
$\bfp = (p_{1}, \dots, p_{n})$
is an optimal strategy for $\calH$, then there exists some
$1 \leq \ell \leq n$
such that
\[
\Support(\bfp)
\, = \,
[\ell] \, .
\]
\end{observation}
\begin{proof}
Assume towards contradiction that there exists some
$1 \leq i \leq n - 1$
such that
$i \notin \Support(\bfp)$
and
$i + 1 \in \Support(\bfp)$,
that is,
$p_{i} = 0$
and
$p_{i + 1} > 0$.
It follows that
\[
\frac{d}{d p} b_{i}(p_{i})
\, = \,
\frac{d}{d p} b_{i}(0)
\, \geq \,
\frac{d}{d p} b_{i + 1}(0)
\, > \,
\frac{d}{d p} b_{i + 1}(p_{i + 1}) \, ,
\]
where the second transition holds by (\ref{equation:rename-cacti}) and
the last transition holds because
$\frac{d}{d p} b_{i + 1}(p)$
is monotonically decreasing in
$[0, 1]$.
This derives a contradiction to Lemma~\ref{lemma:optimality-condition}, thus
establishing the assertion.
\end{proof}

We say that strategy
$\bfp = (p_{1}, \dots, p_{n})$
of the hungry bat instance $\calH$ has \emph{full support} if
$\Support(\bfp) = [n]$.
The instance $\calH$ is said to be \emph{solid} if it admits
a full support optimal strategy.

\begin{lemma} \label{lemma:solid-unique-optimal}
If $\calH$ is solid, then it admits exactly one full support optimal strategy
$\bfp = (p_{1}, \dots, p_{n})$.
Moreover, $\bfp$ can be obtained as the (unique) solution of a system of $n$
independent linear equations over the variables
$p_{1}, \dots, p_{n}$.
\end{lemma}
\begin{proof}
If
$\bfp = (p_{1}, \dots, p_{n})$
is a full support optimal strategy of $\calH$, then
Lemma~\ref{lemma:optimality-condition} ensures that
\[
\frac{d}{d p} b_{i}(p_{i})
\, = \,
\frac{d}{d p} b_{i + 1}(p_{i + 1})
\]
for
$1 \leq i \leq n - 1$.
Using the expressions in (\ref{equation:b-function-derivatives}) for
$\frac{d}{d p} b_{i}(p_{i})$
and
$\frac{d}{d p} b_{i + 1}(p_{i + 1})$,
these
$n - 1$
equations can be rewritten as
\begin{equation} \label{equation:linear-equation-derivatives}
\sqrt{\frac{r_{i + 1} \cdot s_{i + 1}}{1 - s_{i + 1}}} \cdot p_{i}
-
\sqrt{\frac{r_{i} \cdot s_{i}}{1 - s_{i}}} \cdot p_{i + 1}
\, = \,
\frac{s_{i + 1}}{1 - s_{i + 1}} \cdot
\sqrt{\frac{r_{i} \cdot s_{i}}{1 - s_{i}}}
-
\frac{s_{i}}{1 - s_{i}} \cdot
\sqrt{\frac{r_{i + 1} \cdot s_{i + 1}}{1 - s_{i + 1}}}
\end{equation}
for
$1 \leq i \leq n - 1$.
Combined with
\begin{equation} \label{equation:linear-equation-sum-1}
\sum_{i \in [n]} p_{i}
\, = \,
1 \, ,
\end{equation}
we obtain $n$ linear equations in the variables
$p_{1}, \dots, p_{n}$.

Let
$\mathbf{A} = (a_{i, j}) \in \Reals^{n \times n}$
be the coefficient matrix of the aforementioned system of linear equations,
where row
$1 \leq i \leq n - 1$
of $\mathbf{A}$ corresponds to the $i$-th equation in
(\ref{equation:linear-equation-derivatives}) and row $n$ of $\mathbf{A}$
corresponds to the equation in (\ref{equation:linear-equation-sum-1}).
Refer to Figure~\ref{figure:matrix} for an illustration of the signs of
$\mathbf{A}$'s entries.
Consider an arbitrary non-zero (column) vector
$\mathbf{x} = (x_{i}) \in \Reals^{n}$
and let
$\mathbf{x}^{T} \mathbf{A} = \mathbf{c}^{T}$,
where
$\mathbf{c} = (c_{j}) \in \Reals^{n}$.
We prove that $\mathbf{c}$ is not the zero vector, hence $\mathbf{A}$ has full
rank which establishes the uniqueness of $\bfp$.

Assume without loss of generality that
$x_{n} \leq 0$.
Notice that if
$x_{1} = \cdots x_{n - 1} = 0$,
then $x_{n}$ must be strictly negative as
$\mathbf{x} \neq \mathbf{0}$,
hence
$c_{j} < 0$
for every
$1 \leq j \leq n$.
So, assume hereafter that there exists some
$1 \leq i \leq n - 1$
such that
$x_{i} \neq 0$.

Since
$a_{1, 1} > 0$,
$a_{n, 1} > 0$,
and
$a_{i, 1} = 0$
for every
$2 \leq i \leq n - 1$,
it follows that if
$c_{1} = 0$,
then
$x_{1} \geq 0$.
Likewise, since
$a_{n - 1, n} < 0$,
$a_{n, n} > 0$,
and
$a_{i, n} = 0$
for every
$1 \leq i \leq n - 2$,
it follows that if
$c_{n} = 0$,
then
$x_{n - 1} \leq 0$.

Therefore, recalling that
$x_{i} \neq 0$
for some
$1 \leq i \leq n - 1$,
if
$c_{1} = c_{n} = 0$,
then, there must exist some
$2 \leq i \leq n - 1$
such that
(I)
$x_{i - 1} \geq 0$;
(II)
$x_{i} \leq 0$;
and
(III)
at least one of the inequalities in (I) and (II) is strict.
As
$a_{i - 1, i} < 0$,
$a_{i, i} > 0$,
$a_{n, i} > 0$,
and
$a_{i', i} = 0$
for every
$i' \in [n] - \{ i - 1, i, n \}$,
we conclude that
$c_{i} < 0$.
The assertion follows.
\end{proof}

\begin{figure}
{\centering
\includegraphics[width=0.5\textwidth]{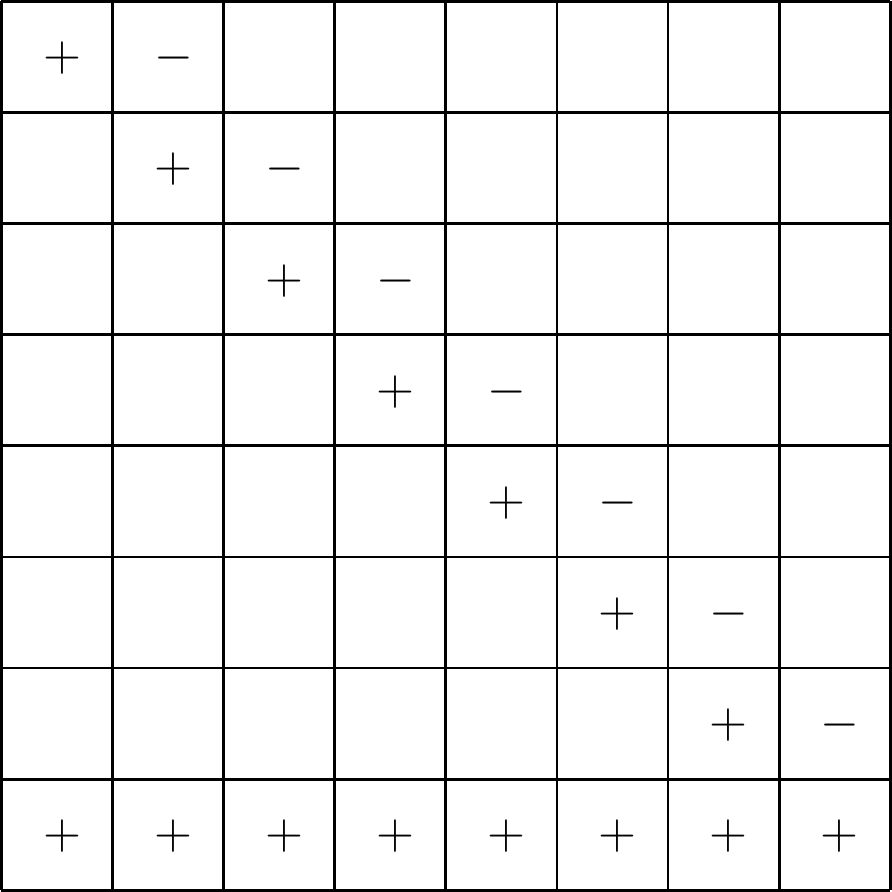}
\par}
\caption{\label{figure:matrix}%
The signs of the entries of matrix $\mathbf{A}$ when
$n = 8$.
Empty cells represent $0$ entries.}
\end{figure}

\begin{corollary} \label{corollary:unique-optimal}
The instance $\calH$ admits a unique optimal solution.
\end{corollary}
\begin{proof}
We first establish the assertion for the case that $\calH$ is solid.
To this end, let
$\bfp = (p_{1}, \dots, p_{n})$
be the unique full support optimal strategy of $\calH$ guaranteed by
Lemma~\ref{lemma:solid-unique-optimal} and assume towards contradiction
that $\calH$ admits another optimal strategy
$\bfq = (q_{1}, \dots, q_{n})$
with
$|\Support(\bfq)| = \ell < n$,
recalling that
$\Support(\bfq) = [\ell]$
by Observation~\ref{observation:prefix}.

Since $\bfq$ assigns all its mass to the cacti in $[\ell]$ while
$\bfp$ assigns positive mass also to the cacti in
$[n] - [\ell]$,
it follows that
\[
\sum_{i \in [\ell]} q_{i}
\, = \,
1
\, > \,
\sum_{i \in [\ell]} p_{i} \, .
\]
By the pigeonhole principle, there exists some
$i \in [n]$
such that
$q_{i} > p_{i}$.
We can now derive a contradiction by developing
\[
\frac{d}{d p} b_{n}(0)
\, > \,
\frac{d}{d p} b_{n}(p_{n})
\, = \,
\frac{d}{d p} b_{i}(p_{i})
\, > \,
\frac{d}{d p} b_{i}(q_{i})
\, \geq \,
\frac{d}{d p} b_{n}(q_{n})
\, = \,
\frac{d}{d p} b_{n}(0) \, ,
\]
where the first and third transitions hold because
$\frac{d}{d p} b_{i}(\cdot)$
is monotonically decreasing in
$[0, 1]$
and the second and fourth transitions follow from
Lemma~\ref{lemma:optimality-condition}.

So, suppose that $\calH$ is not solid and assume towards contradiction that
strategies
$\bfp^{1} = (p^{1}_{1}, \dots, p^{1}_{n})$
and
$\bfp^{2} = (p^{2}_{1}, \dots, p^{2}_{n})$,
$\bfp^{1} \neq \bfp^{2}$,
are optimal for $\calH$.
By Observation~\ref{observation:prefix}, there exist
$1 \leq \ell^{1}, \ell^{2} < n$
such that
$\Support(\bfp^{1}) = [\ell^{1}]$
and
$\Support(\bfp^{2}) = [\ell^{2}]$.

For
$j = 1, 2$,
let
$\calH_{j} =
\langle \ell^{j}, \{ r_{i} \}_{i \in [\ell^{j}]}, \{ s_{i} \}_{i \in [\ell^{j}]}
\rangle$
be the restriction of $\calH$ to the cacti in $[\ell^{j}]$.
Since $\bfp^{j}$ is an optimal strategy for $\calH$ and
$p^{j}_{i} = 0$
for every
$\ell^{j} < i \leq n$, it follows that the
probability vector
$(p^{j}_{1}, \dots, p^{j}_{\ell^{j}})$
is a full support optimal strategy for $\calH_{j}$, which immediately implies
that $\calH_{j}$ is solid.
If
$\ell^{1} = \ell^{2}$,
then we obtain two distinct optimal strategies for the solid instance
$\calH_{1} = \calH_{2}$,
thus deriving a contradiction.

Assume without loss of generality that
$\ell^{1} < \ell^{2}$.
Since $\bfp^{1}$ is an optimal strategy for $\calH$ and
$p^{1}_{i} = 0$
for every
$\ell^{1} < \ell^{2} < i \leq n$, it follows that the
probability vector
$(p^{1}_{1}, \dots, p^{1}_{\ell^{2}})
=
(p^{1}_{1}, \dots, p^{1}_{\ell^{1}}, 0, \dots, 0)$
is an optimal strategy for $\calH_{2}$.
But this probability vector differs from
$(p^{2}_{1}, \dots, p^{2}_{\ell^{2}})$
which is also an optimal strategy for $\calH_{2}$, thus deriving a
contradiction as $\calH_{2}$ is solid.
\end{proof}

We are now ready to develop a computationally efficient algorithm that
constructs the unique optimal strategy for a given hungry bat instance
$\calH =
\langle n, \{ r_{i} \}_{i \in [n]}, \{ s_{i} \}_{i \in [n]}
\rangle$.
For
$1 \leq \ell \leq n$,
let
$\calH^{\ell} =
\langle \ell, \{ r_{i} \}_{i \in [\ell]}, \{ s_{i} \}_{i \in [\ell]}
\rangle$
be the restriction of $\calH$ to the cacti in $[\ell]$.
The instance $\calH^{1}$ is clearly solid, realized by the (degenerated)
strategy that assigns the whole mass to its single cactus.
Assume by induction that the instance $\calH^{\ell}$,
$1 \leq \ell \leq n - 1$,
is solid and that we already hold its (unique) full support optimal strategy
$\bfp^{\ell} = (p^{\ell}_{1}, \dots, p^{\ell}_{\ell})$.
Lemma~\ref{lemma:optimality-condition} ensures that there exists some
$\mu_{\ell}$ such that
$\frac{d}{d p} b_{i}(p^{\ell}_{i}) = \mu_{\ell}$
for every
$i \in [\ell]$.

The algorithm continues by testing whether
\begin{equation} \label{equation:algorithm-test}
\frac{d}{d p} b_{\ell + 1}(0)
\, > \,
\mu_{\ell} \, .
\end{equation}
If (\ref{equation:algorithm-test}) holds, then we conclude that the instance
$\calH^{\ell + 1}$
is solid (see Observation~\ref{observation:test-succeeds}) and apply
Lemma~\ref{lemma:solid-unique-optimal} to construct its full support optimal
strategy
$\bfp^{\ell + 1}$.
Otherwise, the algorithm returns the probability vector
$\bfq = (q_{1}, \dots, q_{n})$
defined by setting
\[
q_{i}
\, = \,
\begin{cases}
p^{\ell}_{i}, & \text{if } 1 \leq i \leq \ell \\
0, & \text{if } \ell < i \leq n
\end{cases}
\]
as the optimal strategy of instance
$\calH = \calH^{n}$.
The algorithm's correctness is derived from the following two observations.

\begin{observation} \label{observation:test-succeeds}
If (\ref{equation:algorithm-test}) holds, then the instance
$\calH^{\ell + 1}$
is solid.
\end{observation}
\begin{proof}
Assume towards contradiction that
$\calH^{\ell + 1}$
is not solid and let
$\bfp' = (p'_{1}, \dots, p'_{\ell}, 0)$
be its optimal strategy.
By definition, the probability vector
$(p'_{1}, \dots, p'_{\ell})$
is an optimal strategy for instance $\calH^{\ell}$.
Since
$\bfp^{\ell} = (p^{\ell}_{1}, \dots, p^{\ell}_{\ell})$
is the unique optimal strategy of $\calH^{\ell}$, it follows that
$p'_{i} = p^{\ell}_{i}$
for every
$1 \leq i \leq \ell$.
This derives a contradiction to Lemma~\ref{lemma:optimality-condition} when
applied to $\bfp'$ and
$\calH^{\ell + 1}$
as (\ref{equation:algorithm-test}) implies that
\[
\frac{d}{d p} b_{\ell + 1}(p'_{\ell + 1})
\, = \,
\frac{d}{d p} b_{\ell + 1}(0)
\, > \,
\mu_{l}
\, = \,
\frac{d}{d p} b_{i}(p^{\ell}_{i})
\, = \,
\frac{d}{d p} b_{i}(p'_{i})
\]
for any
$i \in \Support(\bfp')$.
\end{proof}

\begin{observation} \label{observation:test-fails}
If (\ref{equation:algorithm-test}) does not hold, then $\bfq$ is the optimal
strategy of instance $\calH$.
\end{observation}
\begin{proof}
Since (\ref{equation:algorithm-test}) does not hold, it follows that
\[
\mu_{\ell}
\, \geq \,
\frac{d}{d p} b_{\ell + 1}(0)
\geq
\frac{d}{d p} b_{i}(0)
\, = \,
\frac{d}{d p} b_{i}(q_{i})
\]
for every
$\ell < i \leq n$,
where the second transition is due to (\ref{equation:rename-cacti}).
On the other hand, by the construction of $\bfq$, we know that
\[
\frac{d}{d p} b_{i}(q_{i})
\, = \,
\frac{d}{d p} b_{i}(p^{\ell}_{i})
\, = \,
\mu_{\ell}
\]
for every
$1 \leq i \leq \ell$.
The assertion follows by applying Lemma~\ref{lemma:optimality-condition} to
strategy $\bfq$.
\end{proof}

\section{A Small Core}
\label{section:small-core}
Consider a hungry bat instance
$\calH = \langle n, \{ r_{i} \}_{i \in [n]}, \{ s_{i} \}_{i \in [n]} \rangle$
and fix some
$0 < \varepsilon < 1$.
Let
\[
\sigma
\, = \,
\min_{i \in [n]} s_{i} \, .
\]
In this section, we establish Theorem~\ref{theorem:small-core} by proving that
there exists a strategy $\bfp$ for $\calH$ that satisfies:
(1)
$\bfp$ is supported by a core
$S = \Support(\bfp)$
of size
$|S|
\leq
\frac{2 (1 - \sigma)}{\varepsilon \cdot \sigma}$;
and
(2)
$\mathcal{B}_{\calH}(\bfp)
\geq
(1 - \varepsilon) \cdot \mathcal{B}_{\calH}(\bfp^{*})$,
where $\bfp^{*}$ is the optimal strategy for $\calH$.

Similarly to the assumption made in
Section~\ref{section:solving-optimization-problem} on the order of the cacti
(see \eqref{equation:rename-cacti}), here too we assume that the cacti are
reordered according to their derivatives at
$p = 0$
so that
\[
\chi_{i}
\, = \,
\frac{d}{d p} b_{i}(0)
\, \geq \,
\frac{d}{d p} b_{i + 1}(0)
\, = \,
\chi_{i + 1}
\]
for every
$1 \leq i \leq n - 1$.
Following that, we set the core $S$ so that it includes cacti
$i = 1, \dots, k$,
where
\[
k
\, = \,
\frac{2 \cdot (1 - \sigma)}{\varepsilon \cdot \sigma} \, ,
\]
and prove that there exists a strategy
$\bfp = (p_{1}, \dots, p_{n})$
with support
$\Support(\bfp) = S = [k]$
that satisfies
$\mathcal{B}_{\calH}(\bfp)
\geq
(1 - \varepsilon) \cdot \mathcal{B}_{\calH}(\bfp^{*})$.

Given a strategy
$\bfq = (q_{1}, \dots, q_{n})$,
we refer to
$\sum_{i \in [k]} q_{i}$
and
$\sum_{i \in [n] - [k]} q_{i}$
as the \emph{prefix weight} and \emph{suffix weight} of $\bfq$, respectively,
noting that the two weights sum to $1$. 
Let
$W = \sum_{i \in [n] - [k]} p^{*}_{i}$
denote the suffix weight of $\bfp^{*}$.

The desired strategy $\bfp$ is constructed from the optimal strategy
$\bfp^{*}$ by gradually shifting mass from the suffix weight to the prefix
weight.
For the sake of easier exposition, we present this construction by means of a
(conceptually) iterative procedure that generates a sequence of strategies
$\bfq^{w} = (q^{w}_{1}, \dots, q^{w}_{n})$
for a variable $w$ that ranges from
$w = W$
down to
$w = 0$
in infinitesimally small decrements, where
$\bfq^{W} = \bfp^{*}$
and
$\bfq^{0} = \bfp$.
Throughout this process, we maintain the invariant that the suffix weight of
$\bfq^{w}$ is
$\sum_{i \in [n] - [k]} q^{w}_{i} = w$.

Fix some
$0 < w \leq W$
and suppose that we have already constructed the strategy
$\bfq^{w} = (q^{w}_{1}, \dots, q^{w}_{n})$.
Since the suffix weight of $\bfq^{w}$ is strictly positive, it follows that
there exists some
$j^{-} = j^{-}(w) \in [n] - [k]$
such that
\begin{equation} \label{equation:j-minus}
q^{w}_{j^{-}}
\, > \,
0 \, .
\end{equation}
Furthermore, since the prefix weight of $\bfq^{w}$ is strictly smaller than
$1$, it follows, by the pigeonhole principle, that there exists some
$j^{+} = j^{+}(w) \in [k]$
such that
\begin{equation} \label{equation:j-plus}
q^{w}_{j^{+}}
\, < \,
\frac{1}{k}
\, = \,
\frac{\varepsilon \cdot \sigma}{2 \cdot (1 - \sigma)} \, .
\end{equation}

Let
$d w$
be an infinitesimally small amount of mass.
The strategy
$\bfq^{w - d w} = (q^{w - d w}_{1}, \dots, q^{w - d w}_{n})$
is now generated from $\bfq^{w}$ by transferring
$d w$
mass from cactus $j^{-}$ to cactus $j^{+}$, more formally
\[
q^{w - d w}_{i}
\, = \,
\begin{cases}
q^{w}_{i} + d w, & \text{if } i = j^{+} \\
q^{w}_{i} - d w, & \text{if } i = j^{-} \\
q^{w}_{i}, & \text{otherwise}
\end{cases} \, .
\]
It remains to prove that
$\mathcal{B}_{\calH}(\bfq^{0})
\geq
(1 - \varepsilon) \cdot \mathcal{B}_{\calH}(\bfq^{W})$
or alternatively that
\begin{equation} \label{equation:core-end-to-end}
\mathcal{B}_{\calH}(\bfq^{W}) - \mathcal{B}_{\calH}(\bfq^{0})
\, \leq \,
\varepsilon \cdot \mathcal{B}_{\calH}(\bfq^{W}) \, .
\end{equation}
To this end, we fix some
$0 < w \leq W$
and consider the effect of transferring the
$d w$
mass from cactus
$j^{-} = j^{-}(w) \in [n] - [k]$
to cactus
$j^{+} = j^{+}(w) \in [k]$
when generating strategy
$\bfq^{w - d w}$
from strategy
$\bfq^{w}$.
Recalling that
$\mathcal{B}_{\calH}(\bfp) = \sum_{i \in [n]} b_{i}(p_{i})$
for every strategy
$\bfp = (p_{1}, \dots, p_{n})$,
we conclude that
\[
\mathcal{B}_{\calH}(\bfq^{w})
-
\mathcal{B}_{\calH}(\bfq^{w - d w})
\, = \,
\left(
\frac{d}{d w} b_{j^{-}(w)}(q^{w}_{j^{-}(w)})
-
\frac{d}{d w} b_{j^{+}(w)}(q^{w}_{j^{+}(w)})
\right) \, d w \, .
\]
Therefore,
\[
\mathcal{B}_{\calH}(\bfq^{W}) - \mathcal{B}_{\calH}(\bfq^{0})
\, = \,
\int_{0}^{W}
\left(
\frac{d}{d w} b_{j^{-}(w)}(q^{w}_{j^{-}(w)})
-
\frac{d}{d w} b_{j^{+}(w)}(q^{w}_{j^{+}(w)})
\right) \, d w \, .
\]
The following lemma is pivotal for bounding the right hand side.
 
\begin{lemma} \label{lemma:infinitesimal-difference}
The choice of cacti
$j^{-}(w) \in [n] - [k]$
and
$j^{+}(w) \in [k]$
guarantees that
\[
\frac{d}{d w} b_{j^{-}(w)}(q^{w}_{j^{-}(w)})
-
\frac{d}{d w} b_{j^{+}(w)}(q^{w}_{j^{+}(w)})
\, \leq \,
\varepsilon \cdot \frac{d}{d w} b_{j^{-}(w)}(q^{w}_{j^{-}(w)}) \, .
\]
\end{lemma}

Lemma~\ref{lemma:infinitesimal-difference} will soon be proved, but first, let
us explain how it can be used to establish \eqref{equation:core-end-to-end}.
This lemma implies that
\begin{align*}
\int_{0}^{W}
\left(
\frac{d}{d w} b_{j^{-}(w)}(q^{w}_{j^{-}(w)})
-
\frac{d}{d w} b_{j^{+}(w)}(q^{w}_{j^{+}(w)})
\right) \, d w
\, \leq \, &
\varepsilon \cdot
\int_{0}^{W} \frac{d}{d w} b_{j^{-}(w)}(q^{w}_{j^{-}(w)}) \, d w \\
= \, &
\varepsilon \cdot
\sum_{j \in [n] - [k]} \int_{0}^{q^{W}_{j}}
\frac{d}{d w}
b_{j}(q^{w}_{j})
\, d w \\
= \, &
\varepsilon
\sum_{j \in [n] - [k]} \left(b_{j}(q^{W}_{j})
-
b_{j}(q^{0}_{j})\right) \\
= \, &
\varepsilon \cdot \sum_{j \in [n] - [k]} b_{j}(q^{W}_{j}) \\
\leq \, &
\varepsilon \cdot \mathcal{B}_{\calH}(\bfq^{W}) \, ,
\end{align*}
where the fourth transition holds since, by definition,
$q^{0}_{j} = 0$
for every
$j \in [n] - [k]$.

\begin{proof}[Proof of Lemma~\ref{lemma:infinitesimal-difference}]
Let
$j^{-} = j^{-}(w)$
and
$j^{+} = j^{+}(w)$.
Using \eqref{equation:b-function-derivatives}, we get
\[
\frac{d}{d w} b_{j^{-}}(q^{w}_{j^{-}})
\, = \,
r_{j^{-}} \cdot
\frac{\frac{s_{j^{-}}}{1 - s_{j^{-}}}}
{\left( q^{w}_{j^{-}} + \frac{s_{j^{-}}}{1 - s_{j^{-}}} \right)^{2}}
\, < \,
r_{j^{-}} \cdot
\frac{\frac{s_{j^{-}}}{1 - s_{j^{-}}}}
{\left( \frac{s_{j^{-}}}{1 - s_{j^{-}}} \right)^{2}}
\, = \,
\frac{r_{j^{-}} \cdot (1 - s_{j^{-}})}{s_{j^{-}}} \, ,
\]
where the second transition holds by (\ref{equation:j-minus}),
and
\begin{align*}
\frac{d}{d w} b_{j^{+}}(q^{w}_{j^{+}})
\, = \, &
r_{j^{+}} \cdot
\frac{\frac{s_{j^{+}}}{1 - s_{j^{+}}}}
{\left( q^{w}_{j^{+}} + \frac{s_{j^{+}}}{1 - s_{j^{+}}} \right)^{2}} \\
> \, &
r_{j^{+}} \cdot
\frac{\frac{s_{j^{+}}}{1 - s_{j^{+}}}}
{\left(
\frac{\varepsilon \cdot \sigma}{2 \cdot (1 - \sigma)} +
\frac{s_{j^{+}}}{1 - s_{j^{+}}}
\right)^{2}} \\
\geq \, &
r_{j^{+}} \cdot
\frac{\frac{s_{j^{+}}}{1 - s_{j^{+}}}}
{\left(
\frac{\varepsilon}{2} \cdot \frac{s_{j^{+}}}{1 - s_{j^{+}}} +
\frac{s_{j^{+}}}{1 - s_{j^{+}}}
\right)^{2}} \\
= \, &
r_{j^{+}} \cdot
\frac{\frac{s_{j^{+}}}{1 - s_{j^{+}}}}
{\left(
\left( 1 + \frac{\varepsilon}{2} \right) \cdot \frac{s_{j^{+}}}{1 - s_{j^{+}}}
\right)^{2}} \\
= \, &
\frac{r_{j^{+}} \cdot (1 - s_{j^{+}})}{s_{j^{+}}}
\cdot
\frac{1}{\left( 1 + \frac{\varepsilon}{2} \right)^{2}} \, ,
\end{align*}
where the second transition holds by (\ref{equation:j-plus}) and the third
transition follows from the definition of
$\sigma = \min_{i \in [n]} s_{i}$.
Since
$j^{+} < j^{-}$,
it follows that
\[
\frac{r_{j^{+}} \cdot (1 - s_{j^{+}})}{s_{j^{+}}}
\, \geq \,
\frac{r_{j^{-}} \cdot (1 - s_{j^{-}})}{s_{j^{-}}} \, .
\]
Observing that
$\displaystyle
\frac{1}{\left( 1 + \frac{\varepsilon}{2} \right)^{2}}
>
1 - \varepsilon$,
we conclude that
\[
\frac{d}{d w} b_{j^{+}}(q^{w}_{j^{+}})
\, > \,
(1 - \varepsilon) \cdot \frac{d}{d w} b_{j^{-}}(q^{w}_{j^{-}}) \, ,
\]
thus establishing the assertion.
\end{proof}

We conclude this section by showing that the bound on the core size promised
in Theorem~\ref{theorem:small-core} is asymptotically tight.

\begin{lemma}
Fix some
$0 < \varepsilon < \frac{1}{2}$.
For every
$0 < \sigma < 1$
and for every sufficiently large $n$, there exists a hungry bat instance
$\calH = \langle n, \{ r_{i} \}_{i \in [n]}, \{ s_{i} \}_{i \in [n]} \rangle$
with
$\sigma = \min_{i \in [n]} s_{i}$
such that if $\bfp$ is a strategy for $\calH$ with support size
\[
|\Support(\bfp)|
\, \leq \,
\frac{1 - \sigma}{2 \varepsilon \cdot \sigma} \, ,
\]
then
$\mathcal{B}(\bfp) < (1 - \varepsilon) \cdot \mathcal{B}(\bfp^{*})$,
where $\bfp^{*}$ is the optimal strategy of $\calH$.
\end{lemma}
\begin{proof}
Consider the homogeneous hungry bat instance $\calH^{h}_{n}$
that consists of $n$ identical cacti with parameters
$s_{i} = s$
and
$r_{i} = 1$
for every
$i \in [n]$.
By Lemma~\ref{lemma:optimality-condition}, we know that the uniform
probability vector
$\bfp^{u} = \left( \frac{1}{n}, \dots, \frac{1}{n} \right)$
constitutes the optimal strategy for $\calH^{U}_{n}$, yielding
\[
\mathcal{B}_{\calH^{h}_{n}}(\bfp^{u})
\, = \,
n \cdot \frac{\frac{1}{n}}{\frac{1}{n} + \frac{s}{1 - s}}
\, = \,
\frac{1}{\frac{1}{n} + \frac{s}{1 - s}} \, ,
\]
thus
\[
\lim_{n \rightarrow \infty} \mathcal{B}_{\calH^{h}_{n}}(\bfp^{u})
\, = \,
\frac{1 - s}{s} \, .
\]

On the other hand, by the same argument, the optimal strategy for a core
of size
$k \leq n$
is
$\bfp^{u}_{k} = \left( \frac{1}{k}, \dots, \frac{1}{k}, 0, \dots, 0  \right)$,
yielding
\[
\mathcal{B}_{\calH^{h}_{n}}(\bfp^{u}_{k})
\, = \,
\frac{1}{\frac{1}{k} + \frac{s}{1 - s}} \, .
\]
Therefore, if
$k \leq \frac{1 - s}{2 \varepsilon \cdot s}$,
then
\[
\mathcal{B}_{\calH^{h}_{n}}(\bfp^{u}_{k})
\, \leq \,
\frac{1}
{\frac{2 \varepsilon \cdot s}{1 - s} + \frac{s}{1 - s}}
\, = \,
\frac{1 - s}{(1 + 2 \varepsilon) \cdot s}
\, < \,
\frac{1 - s}{s} \cdot (1 - \varepsilon) \, ,
\]
thus establishing the assertion when $n$ is sufficiently large.
\end{proof}

\bibliographystyle{alpha}
\bibliography{references}

\newcommand{\etalchar}[1]{$^{#1}$}
\begin{thebibliography}{GHE{\etalchar{+}}}

\bibitem[GHE{\etalchar{+}}]{paper}
Aya Goldshtein, Michal Handel, Ofri Eitan, Afrine Bonstein, Talia Shaler, Simon
  Collet, Stefan Greif, Rodrigo~A. Medellin, Yuval Emek, Amos Korman, and Yossi
  Yovel.
\newblock Decision-making in the wild: Reinforcement learning enables
  resource-partitioning in foraging bats.
\newblock To appear.

\end{thebibliography}

\end{document}